\documentclass{article}
\usepackage{amsmath,amssymb,mathrsfs,amsthm}

\theoremstyle{plain}
\newtheorem{thm}{Theorem} 
\newtheorem{prop}{Proposition}[section]
\newtheorem{lemma}[prop]{Lemma}

\theoremstyle{definition} 
\newtheorem{asm}{Assumption} 

\theoremstyle{remark}
\newtheorem{remark}[prop]{Remark}

\newcommand{\RR}{\mathbb{R}}
\newcommand{\CC}{\mathbb{C}}

\newcommand{\Cinf}{C^{\infty}}
\newcommand{\ve}{\varepsilon}
\newcommand{\supp}{\mathop{\mathrm{supp}}}
\newcommand{\RE}{\mathrm{Re}\,}
\newcommand{\IM}{\mathrm{Im}\,}

\begin{document}

\title{Many-body Stark resonances by the complex absorbing potential method}
\author{Kentaro Kameoka}

\date{}

\maketitle

\begin{abstract}
The resonances of many-body Stark Hamiltonians are characterized by the complex absorbing 
potential method. Namely, the resonances are shown to be the limit points of complex discrete eigenvalues 
of many-body Stark Hamiltonians with quadratic complex potential when the coefficient of the 
complex potential tends to zero.
\end{abstract}

\section{Introduction}
In this paper, we discuss the complex absorbing potential method for resonances of 
the following many-body Stark Hamiltonians
\[P
=\sum_j(-\frac{1}{2m_j}\Delta_j+V_j(x_j)+q_j x_j^1)+\sum_{j<k}V_{jk}(x_j-x_k)
\]
on $L^2(\RR^{d N})$. 
In this paper, the indices $j, k$ run through $1$ to $N$ unless otherwise stated.
Here $d$ and $N$ are positive integers 
and $m_j$ and $q_j$ are positive real numbers. 
These are fixed in this paper. 
We assume $\min_j q_j=1$ for simplicity. 
The coordinate is written as $x_j=(x_j^1, x_j')\in \RR^{d}=\RR\times\RR^{d-1}$ and $\Delta_j$ is the 
Laplacian with respect to $x_j$. We also set $x=(x_1,\cdots, x_N)\in \RR^{d N}$. 
The potentials $V_j$ and $V_{jk}$ are real-valued functions on $\RR^{d}$ which satisfy the following assumption. 
 
\begin{asm}\label{asm-1}
There exist $R_0>0$, $\delta_0>0$ and $c_0>0$ such that the following hold. 
(i) For any $1\le j \le N$, 
$V_j\in \Cinf(\RR^{d}; \RR)$ has an analytic continuation to the region 
$\{y\in \CC^{d}|\, |\RE y|>R_0, |\IM y|<\delta_0\}$ and $\lim_{|y|\to \infty}\partial V_j(y)=0$ there. 
(ii) For any $1\le j <k\le N$, 
$V_{jk}\in\Cinf(\RR^{d}; \RR)$ is bounded and has an analytic continuation to the region 
$\{y\in \CC^{d}|\, |\IM y|<\min\{\delta_0, c_0|\RE y|\}\}$. Moreover,  
$\partial V_{jk}$ is bounded and $\lim_{|y|\to \infty}\partial V_{jk}(y)=0$ in that region. 
\end{asm}

Under this assumption, $P$ is essentially self-adjoint on $\Cinf_c(\RR^{d N})$ by the Faris-Lavine theorem.
Moreover, we can define the resonances of $P$ as follows (see Section~2 for details).  
We set $R_{+}(z)=(z-P)^{-1}$ for $\IM z>0$.
Then for any non-zero $\chi_1, \chi_2\in \Cinf_c(\RR^{d N})$, the cutoff resolvent $\chi_1 R_{+}(z)\chi_2$ has 
a meromorphic continuation from the upper half plane to the region $\{z\in \CC|\, \IM z>-\delta_0\}$. 
We define the multiplicity $m_z$ of resonance by
\[m_z=\mathrm{rank}\frac{1}{2\pi i}\oint_z \chi_1 R_+(z) \chi_2 dz, \]
which is independent of the choice of $\chi_1, \chi_2$. We call $z$ a resonance of $P$ if $m_z>0$ and 
the set of resonances is denoted by $\mathrm{Res}(P)$.

We set 
\[
P_{\ve}=P-i\ve x^2
\]
for $\ve\ge 0$, where $x^2=\sum_j x_j^2$. 
Then $P_{\ve}$ for $\ve>0$ with $D(P_{\ve})=D(-\Delta)\cap D(x^2)$ 
has purely discrete spectrum. 
Then our main theorem is the following.

\begin{thm}\label{thm-1}
Under Assumption~\ref{asm-1},  
\[
\lim_{\ve \to +0}\sigma_d (P_{\ve})=\mathrm{Res}(P) 
\]
in $\{z\in \CC|\, \IM z>-\delta_0\}$ including multiplicities. 

More precisely, for any $z$ with $\IM z>-\delta_0$,  
there exists $\gamma_0>0$ such that 
for any $0<\gamma<\gamma_0$ there exists $\ve_0>0$ 
such that for any $0<\ve<\ve_0$, the number of eigenvalues of $P_{\ve}$ in 
$\{\zeta \in \mathbb{C}|\, |\zeta-z|\le \gamma\}$ counted with their algebraic multiplicities 
coincides with $m_z$. 
\end{thm}

\begin{remark}
Our proof of Theorem~\ref{thm-1} preserves the symmetry with respect to permutations of particles 
if $V_j$, $V_{jk}$, $m_j$ and $q_j$ are independent of indices. 
Thus, for instance, our result holds true for fermions with spins. 
\end{remark}

The complex absorbing potential method was introduced in physical chemistry \cite{RM}, \cite{SeMi}. 
The mathematical justification was given in \cite{Z2} for compactly supported potentials. 
This was extended to several settings in \cite{K2}, \cite{KN}, \cite{X1}, \cite{X2}, \cite{X3}. 
Analogous results hold for Policott-Ruelle resonances \cite{DZ}
and for zeroth-order pseudodifferential operators \cite{GZ}. 

Resonances of many-body Stark Hamiltonians were studied in \cite{HeSi}, \cite{Si}, \cite{W}. 
We study resonances of many-body Stark Hamiltonians by the complex distortion outside a cone, 
which generalizes the one-body case introduced in \cite{K1}. 

For one-body Stark Hamiltonians, the complex absorbing potential method was proved in \cite{K2}. 
In \cite{K2}, we used the fact that the Hamiltonian is a relatively compact perturbation of 
the free Stark Hamiltonian, which has no resonance (see Remark~\ref{remark-one-body}). 
Since this is not true in the many-body case, we need a different strategy.

This paper is organized as follows. 
In Section~2, we discuss the complex distortion outside a cone for many-body Stark Hamiltonians. 
In Section~3, we prove Theorem~\ref{thm-1} assuming our main estimate Proposition~\ref{main-estimate}. 
In Section~4, we prove Proposition~\ref{main-estimate}.

\section{Complex distortion}
In this section, we discuss the complex distortion outside a cone for many-body Stark Hamiltonians, 
which was introduced in \cite{K1} in the one-body case. 
We take $\kappa > 1$, $\rho > 1$ and set $C(\kappa, \rho)=\{y\in \RR^{d}|\, |y'|\le \kappa (y_1+\rho)\}$. 
Take a convex set $\widetilde{C}(\kappa, \rho)$ which has a smooth boundary 
such that $\widetilde{C}(\kappa, \rho)$ is rotationally symmetric with respect to $y'$ and 
$\widetilde{C}(\kappa, \rho)=C(\kappa,\rho)$ in $y_1>-\rho+1$. 
We also take $\phi \in C_c^{\infty}(\RR^{d})$, where $\phi \ge 0$ and $\int\phi(y)dy=1$. 
We set $\phi_{\tau}(y)=\tau^{-d}\phi(y/\tau)$ for $\tau>1$. 
We define $F=-(1+\kappa^{-2})^{1/2}\mathrm{dist}\bigl(\bullet, \widetilde{C}(\kappa, \rho)\bigr)\ast \phi_{\tau}$. 
We also set $v(y)=(v_1(y), \dots, v_{d}(y))=\partial F(y)\in C_b^{\infty}(\RR^{d}; \RR^{d})$ 
(the space of bounded smooth functions with bounded derivatives is denoted by $C_b^{\infty}$). 
Note that the derivatives of $v$ and thus the Lipschitz constant of $v$ are 
sufficiently small if we take $\tau \gg 1$. 
We next set $\Phi_{\theta}(y)=y+\theta v(y)$ for real $\theta$ with $|\theta|\ll 1$, 
which is a diffeomorphism. 
The many-body case is defined as $\Phi_{N, \theta}(x)=(\Phi_{\theta}(x_1), \dots, \Phi_{\theta}(x_N))$. 
We also write $\Phi_{N, \theta}(x)=x_{\theta}=(x_{1, \theta}, \dots, x_{N, \theta})$. 
We set $U_{\theta}f(x)=\left(\mathrm{det}\Phi_{N, \theta}'(x)\right)^{1/2}f(\Phi_{N, \theta}(x))$, 
which is unitary on $L^2(\RR^{d N})$. 
We define the distorted operator $P_{\theta}=U_{\theta}P U_{\theta}^{-1}$. 
\begin{remark}
The notation $P_{\theta}$ and $P_{j, \theta}$ below 
are valid only in this section and thus do not cause confusion with $P_{\ve}$ or $P_{\ve, \theta}$. 
In the other sections, we write $P_{0, \theta}$ instead of $P_{\theta}$. 
Alternatively, we write $P_{\theta=0.1}$ or $P_{\ve=0.1}$ if we consider a specific value. 
\end{remark}
Then we have 
\[
P_{\theta}=\sum_j P_{j, \theta}+\sum_{j<k} V_{jk}(x_{j, \theta}-x_{k, \theta}), 
\]
where each one-body part has the following form
\begin{align*}
P_{j, \theta}
&=-\frac{1}{2m_j}\partial_j \Phi_{\theta}'(x_j)^{-2}\hspace{0.1cm} ^t\!\partial_j
+r_{\theta}(x_j)+q_j x_j^1+q_j\theta v_1(x_j)+V_j(x_{j, \theta}). 
\end{align*}
Here $\partial_j=(\frac{\partial}{\partial x_j^1}, \dots, \frac{\partial}{\partial x_j^{d}})$ is a row vector and 
$\Phi_{\theta}'$ is the Jacobi matrix. 
For $\IM\theta<0$, we have $\IM\Phi_{\theta}'(y)^{-2}\le 0$ and $\RE\Phi_{\theta}'(y)^{-2}\ge 0$ 
in the form sense. 
We have $v_1(y)\ge 0$ everywhere and $v_1(y)\ge 1$ if $\mathrm{dist}\bigl(y, \widetilde{C}(\kappa, \rho)\bigr)\gg 1$. 
The function $r_{\theta}(y)\in\Cinf_b(\RR^{d})$ is small in $\Cinf_b(\RR^{d})$ if 
$\tau \gg 1$ since it involves higher derivatives of $v$ 
(see~\cite[Section~2]{K1} for these facts and the concrete expression of $r_{\theta}$). 

\begin{remark}
Obviously, it is enough to assume that $V_j$ is analytic outside $C(\kappa, \rho)$ for some $\kappa>1$ and $\rho>1$. 
We wrote a stronger assumption for simplicity in Section~1.  
\end{remark}

We interpret $P_{\theta}$ as the closure of the operator $P_{\theta}$ defined on $\Cinf_c$. 
We take $\tau>1$ large enough so that $V_{jk}(x_{j, \theta}-x_{k, \theta})$ 
is analytic with respect to $\theta$ (see~Assumption~\ref{asm-1}). 
We also take $\kappa>1$, $\rho>1$ large enough so that $V_j(x_{j, \theta})$ is analytic with respect to $\theta$.  
We first prove the basic operator-theoretic properties of $P_{\theta}$ using the theory of semigroup.

\begin{prop}
The operator $P_{\theta}$ is an analytic family of closed operators in the sense of Kato 
for $\theta$ with $|\IM\theta|<\delta_0(1+\kappa^{-2})^{-1/2}$ and $|\RE\theta|$ small 
if $\kappa>1$, $\rho>1$ and $\tau>1$ in the definition of $P_{\theta}$ are large enough. 
We also have $P_{\theta}^*=P_{\bar{\theta}}$. 
\end{prop}

\begin{proof}
This proposition was proved for the one-body case in \cite[Section~2]{K1}. 
Although we assumed that $V_j$ is decaying in \cite[Section~2]{K1}, the same argument applies 
since $V_j(x_{j, \theta})-V_j(x_{j, \theta'})$ is bounded (recall that $\partial V_j$ is bounded). 
Since $\IM(u, (P_{j, \theta}-iC-i\mu) u)\le -\mu \|u\|^2$ for some $C>0$ and any $\mu>0$ with the 
similar estimate for the adjoint, 
we can consider the semigroup $e^{-itP_{j, \theta}}$, $t\ge 0$, by the Hille-Yosida theorem. 
Then we can consider the semigroup $e^{-itP_{1, \theta}}\cdots e^{-itP_{N, \theta}}$. 
Since this semigroup preserves $D(P_{1, \theta})\otimes \cdots \otimes D(P_{N, \theta})$, 
it is a core for the generator $\sum_{j=1}^N P_{j, \theta}$ of this semigroup. 
Since $\Cinf_c(\RR^{d})$ is a core for $P_{j, \theta}$, we conclude that 
the generator $\sum_{j=1}^N P_{j, \theta}$ coincides with the closure of its restriction to $\Cinf_c(\RR^{d N})$. 
By the formula $e^{-itP_{j, \theta}}=\lim_{n\to \infty}(1+\frac{it}{n}P_{j, \theta})^{-n}$ and the 
analyticity of $(z-P_{j, \theta})^{-1}$ with respect to $\theta$, we see that $e^{-itP_{j, \theta}}$ is 
analytic with respect to $\theta$. 
Thus $(z-\sum_{j=1}^N P_{j, \theta})^{-1}=-i\int_0^{\infty}e^{itz}e^{-itP_{1, \theta}}\cdots e^{-itP_{N, \theta}}dt$ 
is analytic with respect to $\theta$ for $\IM z\gg 1$. 
Since the generator of the adjoint semigroup is the adjoint of the original semigroup by the above relation between 
the resolvent and semigroup, we see $(\sum_{j=1}^N P_{j, \theta})^*=\sum_{j=1}^N P_{j, \bar{\theta}}$. 
We see that $\sum_{j<k}^N V_{jk}(x_{j, \theta}-x_{k, \theta})$ is an analytic family of bounded operators which is
self-adjoint for real $\theta$ under our assumption. 
Note that $|v(x_j)-v(x_k)|\le \max\{|v(x_j)|, |v(x_k)|\}\le (1+\kappa^{-2})^{1/2}$ since two vectors 
$v(x_j)$ and $v(x_k)$ have almost the same directions for $\kappa \gg 1$. 
Thus $|\IM (x_{j, \theta}-x_{k, \theta})|<\delta_0$ for our $\theta$. 
Thus we see that $P_{\theta}^*=P_{\bar{\theta}}$. 
The analyticity of $P_{\theta}$ with respect to $\theta$ follows from that of 
$\sum_{j=1}^N P_{j, \theta}$ and the Kato's definition 
of the analyticity (see~\cite[Section~VII, 2]{Ka}) or the resolvent equation. 
\end{proof}

We next prove the discreteness of $\sigma (P_{\theta})$ and the fact that $\mathrm{Res}(P)=\sigma (P_{\theta})$ 
using Proposition~\ref{main-estimate} in the next section. 
\begin{prop}
For any $0<\delta_1<\delta<\delta_0$, 
the operator $P_{\theta}$ with $\theta=-i\delta$ has purely discrete spectrum in 
$\{z|\,\IM z>-\delta_1\}$ if $\kappa>1$, $\rho>1$ and $\tau>1$ are large enough. 
Resonances of $P$ coincide with discrete eigenvalues of $P_{\theta}$ including multiplicities there.  
\end{prop}
\begin{proof}
By Proposition~\ref{main-estimate}, there exists a bounded operator $K$ such that 
$(P_{\theta}-z-iK)^{-1}$ exists and $K(P_{\theta}-z-iK)^{-1}$ is compact. 
Thus $P_{\theta}-z=(1+iK(P_{\theta}-z-iK)^{-1})(P_{\theta}-z-iK)$ is a Fredholm operator with index zero. 
Since $-\IM (u, (P_{\theta}-z)u)\gtrsim\|u\|^2$ for $\IM z\gg 1$, we conclude that $(P_{\theta}-z)^{-1}$ is 
meromorphic in $\{z|\, \IM z>\delta_1\}$ with finite rank poles by the analytic Fredholm theory. 
Thus $P_{\theta}$ has purely discrete spectrum there. 
Once we established this fact, standard arguments in resonance theory (see for instance~\cite[Section~2]{K1}) imply that 
the cutoff resolvent $\chi_1 R_{+}(z)\chi_2$ has a meromorphic continuation to 
$\{z|\,\IM z>-\delta_1\}$, 
the multiplicity $m_z$ of resonance $z$ is independent of $\chi_1, \chi_2$, 
and resonances of $P$ coincide with discrete eigenvalues of $P_{\theta}$ there including multiplicities. 
\end{proof}

\section{Complex absorbing potential method}
In this section, we prove Theorem~\ref{thm-1} while the proof of the main estimate (Proposition~\ref{main-estimate}) is 
postponed to the next section. 
In the proof of Theorem~\ref{thm-1}, we also deform $P_{\ve}$ 
to obtain $P_{\ve, \theta}=U_{\theta}P_{\ve}U_{\theta}^{-1}$. 
The basic properties of $P_{0, \theta}$ was discussed in the previous section. 
As in the one-body case (see~\cite{K2}), 
it is easier to see analogous properties for $P_{\ve, \theta}$ with $\ve>0$. 
Namely, $P_{\ve, \theta}$ with $\ve>0$ with the domain $D(P_{\ve, \theta})=D(-\Delta)\cap D(x^2)$ 
is an analytic family of closed operators for $\theta$ with 
$|\mathrm{Im}\theta|<\delta_0(1+\kappa^{-2})^{-1/2}$ and $|\mathrm{Re}\theta|\ll 1$. 
The space $\Cinf_c(\RR^{d N})$ is a core for $P_{\ve, \theta}$. 
Moreover, we have $P_{\ve, \theta}^*=P_{-\ve, \bar{\theta}}$.
The deformed operator $P_{\ve, \theta}$ with $\ve>0$ has purely discrete spectrum on the whole complex plane.  
The eigenvalues of $P_{\ve, \theta}$ with $\ve>0$ are independent of $\theta$ including multiplicities. 
Thus it is enough to prove $\lim_{\ve \to +0}\sigma_d(P_{\ve, \theta})=\sigma_d(P_{0, \theta})$. 

In the rest of this paper, we fix 
$\Omega \Subset \{z|\,\IM z>-\delta_0\}$. 
We fix $0<\delta_1<\delta_0$ such that $\Omega \Subset \{z|\,\IM z>-\delta_1\}$. 
We also fix $0<\delta_1<\delta<\delta_0$ and set $\theta=-i\delta$. 
Then $P_{\ve, \theta}$ with $\ve \ge 0$ is well-defined for this $\theta$ 
if we take $\tau>1$, $\rho>1$ and $\kappa>1$ large enough in the definition of the complex distortion. 

We take $0\le \chi \in \Cinf_c(\RR^{d N})$ such that $\chi(x)=1$ if $|x|<2$ and $\chi(x)=0$ for $|x|>3$. 
We set $K^R(x)=R\chi(x/R)$.  
The main estimate in this paper is the following.

\begin{prop}\label{main-estimate}
If $\tau>1$, $\rho>1$ and $\kappa>1$ in the definition of $P_{\theta}$ are fixed large enough, 
there exists $R>1$ and $C>0$ such that 
\[
\|(P_{\ve, \theta}-z-iK^R)^{-1}\|\le C
\]
for $0 \le \ve \ll 1$ and $z\in \Omega$. 
Moreover, for any $0 \le k\le 2$, $M>1$, $0 \le \ve \ll 1$ and $z\in \Omega$, 
\[
\|(P_{\ve, \theta}-z-iK^R)^{-1}\|_{L^2\to H^k(\{|x|<M\})}\le CM^{k/2}.
\]
\end{prop}
The proof of this proposition is explained in the next section. 
Instead, we give the proof of the main theorem using this proposition. 

\begin{proof}[Proof of Theorem~\ref{thm-1}]
We fix $R>1$ in Proposition~\ref{main-estimate} and write $K=K^R$. 
We also take  $0 \le \ve \ll 1$ and $z\in \Omega$. 
By Proposition~\ref{main-estimate}, we can write 
\[
P_{\ve, \theta}-z=(1+iK(P_{\ve, \theta}-z-iK)^{-1})(P_{\ve, \theta}-z-iK). 
\]
By Proposition~\ref{main-estimate} with $0<k\le 2$, the operator $K(P_{\ve, \theta}-z-iK)^{-1}$ is compact. 
Thus we can apply the theory of Fredholm operators. 
By the same arguments based on the Gohberg-Sigal-Rouch\'{e} theorem as in \cite{Z2} (see also~\cite{K2}, \cite{KN}), 
it is enough to prove 
\[
\lim_{\ve \to +0} \|K(P_{\ve, \theta}-z-iK)^{-1}-K(P_{0, \theta}-z-iK)^{-1}\|_{L^2 \to L^2}=0
\]
uniformly in $\Omega$. 

We set $\chi^M (x)=\chi(x/M)$, where $\chi$ is a cutoff near $0$ as above. 
We write 
\[
(P_{\ve, \theta}-z-iK)\left(\chi^M (P_{0, \theta}-z-iK)^{-1}+(1-\chi^M)(P_{\ve, \theta}-z-iK)^{-1}\right)=1+E_{\ve, z}^M. 
\]
Then a simple calculation shows that 
\begin{align*}
E_{\ve, z}^M=[P_{\ve, \theta}, \chi^M]&(P_{0, \theta}-z-iK)^{-1}
-[P_{\ve, \theta}, \chi^M](P_{\ve, \theta}-z-iK)^{-1}\\
&-\chi^Mi\ve x_{\theta}^2(P_{0, \theta}-z-iK)^{-1}. 
\end{align*}
Then Proposition~\ref{main-estimate} with $k=1$ (with $M$ replaced with $3M$) implies that 
\[
\|E_{\ve, z}^M\|_{L^2 \to L^2}=\mathcal{O}(M^{-1/2})+\mathcal{O}_M(\ve).
\]
If we take $M$ large and then take $\ve>0$ small, we have 
$\|E_{\ve, z}^M\|_{L^2\to L^2}<1/2$ for $z \in \Omega$. 
Thus the Neumann series argument implies that 
\begin{align*}
&(P_{\ve, \theta}-z-iK)^{-1}\\
&=\left(\chi^M (P_{0, \theta}-z-iK)^{-1}+(1-\chi^M)(P_{\ve, \theta}-z-iK)^{-1}\right)(1+E_{\ve, z}^M)^{-1}. 
\end{align*}
Then we have 
\[
K(P_{\ve, \theta}-z-iK)^{-1}-K(P_{0, \theta}-z-iK)^{-1}=K(P_{0, \theta}-z-iK)^{-1}((1+E_{\ve, z}^M)^{-1}-1), 
\]
where we used the fact that $K(1-\chi^M)=0$ for $M \gg 1$.  
If we take large $M>1$ and then take small $\ve>0$, 
the operator $(1+E_{\ve, z}^M)^{-1}$ is close to the identity operator in the 
operator norm by the above estimate of $\|E_{\ve, z}^M\|_{L^2 \to L^2}$. 
This completes the proof of Theorem~\ref{thm-1}.
\end{proof}

\begin{remark}\label{remark-one-body}
In one-body case~\cite{K2}, we 
wrote $P_{\ve, \theta}-z=(1+V_{\theta}(P_{\ve, \theta}-V_{\theta}-z)^{-1})(P_{\ve, \theta}-V_{\theta}-z)$, 
where $V_{\theta}=V(x_{\theta})$. 
The existence of $(P_{0, \theta}-V_{\theta}-z)^{-1}$ corresponds to 
the fact that the free Stark Hamiltonian has no resonance. 
To guarantee the compactness of $V_{\theta}(P_{0, \theta}-V_{\theta}-z)^{-1}$, we assumed that $V$ is decaying. 
By applying the strategy of the present paper with $K$ replaced with $K-iV_{\mathrm{sing}}$ 
($V_{\mathrm{sing}}$ is the singular part of $V$), 
the result in \cite{K2} holds true even if 
we replace the condition that $V$ is decaying with the condition that $\partial V$ is decaying. 
Note that the result in \cite{K2} for one-body case is not a special case of our Theorem~\ref{thm-1} since 
the potentials with local singularities were allowed there. 
Our strategy also applies to the case where the second-order part has variable coefficients under suitable assumptions. 
While it is true even in the many-body case, we only treat the constant coefficient case for simplicity. 
\end{remark}

\section{Proof of the main estimate}
In this section, we prove Proposition~\ref{main-estimate}. 
We use the same notation as in the previous section. 
We take $R>1$, $M>1$, $0 \le \ve <1$ and $z \in \Omega$. 
We take $\tilde{\chi}\in \Cinf(\RR^{d})$ depending only on $y_1$ 
such that $\tilde{\chi}(y)=1$ for $y_1>-1$ and $\tilde{\chi}(y)=0$ for $y_1<-2$. 
We set $\tilde{\chi}^M(y)=\tilde{\chi}(y/M)$. 
We then set 
\[
A=A^{R, M}_{\ve, z}=P_{\ve, \theta}-z-iK^R+\sum_{j=1}^N 2 q_j |x^1_j|(1-\tilde{\chi}^M(x_j))-iM.
\]

In this section, we use the theory of pseudodifferential operators (see~\cite{Z}). 
We recall the notation $\langle x \rangle=(1+x^2)^{1/2}$ and 
\[
S(m)=\{a\in \Cinf| |\partial^{\alpha}_x \partial^{\beta}_{\xi}a(x, \xi)|\le C_{\alpha \beta}m(x, \xi)\} 
\]
with the corresponding class of pseudodifferential operators denoted by $\mathrm{Op}S(m)$. 

\begin{lemma}\label{main-lemma}
There exists $M_0>1$ such that for any $0 \le k\le 2$, 
\[
A^{-1}=\mathcal{O}(M^{-1+k/2}) \,\,\,\,\mathrm{in}\,\,\, \mathrm{Op}S(\langle \xi \rangle^{-k})
\] 
for $M>M_0$ uniformly for $R>1$, $z\in \Omega$ and $0\le \ve <1$. 
\end{lemma}

\begin{proof}
The following arguments are uniform with respect to $R>1$, $z \in \Omega$ and $0 \le \ve <1$, 
which we do not write explicitly. 
We set $A=a(x, D; M)$. 
We estimate $\partial^{\alpha}_x \partial^{\beta}_{\xi}a(x, \xi)^{-1}$. 
We first see that  
\begin{equation}\label{estimate-1}
\left|\frac{\langle \xi \rangle^k}
{a(x, \xi; M)}\right|\lesssim M^{-1+k/2},
\end{equation}
which is seen by estimating it separately 
for $|\xi|/M^{1/2}\gg 1$ and $|\xi|/M^{1/2}\lesssim 1$. 
We used $x_j\cdot v(x_j)\le 0$ (see~\cite[Lemma~2.1]{K2}), which implies $\RE(-i\ve x_{j, \theta}^2)\ge 0$. 
The inequality~\eqref{estimate-1} implies that
\begin{equation}\label{estimate-2}
\left|\frac{\partial_{\xi} a(x, \xi; M)}
{a(x, \xi; M)}\right|
\lesssim \left|\frac{\langle \xi \rangle }
{a(x, \xi; M)}\right| \lesssim M^{-1/2}. 
\end{equation}
We next claim that 
\begin{equation}\label{estimate-3}
\left|\frac{\partial_x a(x, \xi; M)}
{a(x, \xi; M)}\right|
\lesssim \frac{1+\xi^2+\ve x^2}
{|a(x, \xi; M)|}\lesssim 1, 
\end{equation} 
which is also seen by estimating it separately 
for $|\xi|/M^{1/2}\gg 1$ and $|\xi|/M^{1/2}\lesssim 1$.

The inequalities \eqref{estimate-1}, \eqref{estimate-2} and \eqref{estimate-3} imply 
that $a(x, \xi; M)^{-1}$ is $\mathcal{O}(M^{-1+k/2})$ in $S(\langle \xi \rangle^k)$ when $M\to \infty$. 
This implies that $\partial_x a(x, \xi; M)^{-1}$ is $\mathcal{O}(M^{-1/2})$ in $S(\langle \xi \rangle^{-1})$. 
We also see that $\partial_{\xi} a(x, \xi; M)^{-1}$ is $\mathcal{O}(M^{-1/2})$ in $S((1+\xi^2+\ve x^2)^{-1})$.  
This can be seen from 
\[
\left|\frac{1+\xi^2+\ve x^2}
{a(x, \xi; M)}\cdot \frac{\langle \xi \rangle}
{a(x, \xi; M)}\right| \lesssim M^{-1/2} 
\]
and $(\partial_x a)/a, (\partial_{\xi}a)/a \lesssim 1$. 
Since $\partial_x a \in S(1+\xi^2+\ve x^2)$ and $\partial_{\xi}a \in S(\langle \xi \rangle)$, 
a standard argument for pseudodifferential operators (see~\cite[Chapter~5]{Z})
implies that 
the symbols of $a(x, D; M)a^{-1}(x, D; M)-1$ and $a^{-1}(x, D; M)a(x, D; M)-1$ are 
$\mathcal{O}(M^{-1/2})$ in $S(1)$. 
Thus the Neumann series argument and the Beals's theorem complete the proof.
\end{proof}

Using this lemma, we prove the proposition in the previous section. 

\begin{proof}[Proof of Proposition~\ref{main-estimate}]
The following arguments are uniform with respect to $z \in \Omega$ and $0 \le \ve \ll 1$, 
which we do not write explicitly. 
We first prove the existence of $(P_{\ve, \theta}-z-iK^R)^{-1}$. 
For that, we first take a many-body partition of unity. 
We set $q=\max_j q_j$ and take $0<c_1 \ll 1$.  
We take $\chi_1, \chi_2, \chi_3\in \Cinf_b(\RR^{d N})$ such that $0\le \chi_l\le 1$ 
for $l=1,2,3$ and $\sum_{l=1}^3 \chi_l=1$ with the following properties.

\noindent  
(i) If $\chi_1(x)>0$, then 
 $x_j\in \{y\in \RR^{d}|y_1\le -c_1/4qN\}\cup \{y\in \RR^{d}|y_1\le c_1, |y'|\ge 1/2\}$ for some $1\le j\le N$. 

\noindent 
(ii) If $\chi_2(x)>0$, then $x_j\in \{y\in \RR^{d}|y_1\ge -c_1/2qN, |y'|\le 1\}\cup \{y\in \RR^{d}|y_1\ge 3c_1/4\}$ 
for any $1\le j\le N$. 
Moreover, $x_j\in \{y\in \RR^{d}|y_1\ge 3c_1/4\}$ for some $1\le j\le N$. 

\noindent 
(iii) If $\chi_3(x)>0$, then $x_j\in \{y\in \RR^{d}|-c_1/2qN\le y_1\le c_1, |y'|\le 1\}$ for any $1\le j\le N$. 

\noindent 
Such a partition of unity can be easily constructed by summing products of a one-body partition unity. 
We define a rescaled partition of unity by $\chi_l^R(x)=\chi_l(x/R)$. 

We next estimate $\|(P_{\ve, \theta}-z-iK^R)\chi_l^R u\|$ from below for $u\in \Cinf_c(\RR^{d N})$. 
We first claim
\[
-\IM(\chi_1^R u, (P_{\ve, \theta}-z-iK^R)\chi_1^R u)\gtrsim\|\chi_1^R u\|^2 
\]
by our complex distortion. 
Recall that three parameters $\kappa>1$, $\rho>1$ and $\tau>1$ appeared in the definition of the complex distortion. 
We note that $|\IM V_{jk}(x_{j, \theta}-x_{k, \theta})|$ is sufficiently small for 
$\tau \gg 1$.
In fact, it is small when $|x_j-x_k|\gg 1$ by the decay of $\IM V_{jk}$.  
It is small when $|x_j-x_k|\lesssim 1$ by the smallness of $|\IM(x_j-x_k+\theta (v(x_j)-v(x_k)))|$, which follows 
from the smallness of the Lipschitz constant of the vector field $v$.  
Thus we can neglect the pair potential terms. 
The contribution from $V_j(x_{j, \theta})$ is also negligible if $\rho>1$ is large enough depending on $\tau \gg 1$. 
We recall from Section~2 that $r_{\theta}(x_j)$ is also negligible if we take $\tau \gg 1$. 
We assumed that $0\le \ve \ll 1$ to estimate the constant term of $-i\ve x_{\theta}^2$.

We also have 
\[
\RE(\chi_2^R u, (P_{\ve, \theta}-z-iK^R)\chi_2^R u)\gtrsim R\|\chi_2^R u\|^2
\]
by the positive divergence of Stark potential $q_j x_j^1$ for some $j$ in the condition (ii) above 
(the contribution from negative divergence of other Stark potentials is smaller). 
We also have 
\[
-\IM(\chi_3^R u, (P_{\ve, \theta}-z-iK^R)\chi_3^R u)\gtrsim R\|\chi_3^R u\|^2
\]
by the $-iK^R$ term. 

We then have
\begin{align*}
\|u\|&\le \sum_l \|\chi_l^R u\| \lesssim \sum_l\|(P_{\ve, \theta}-z-iK^R)\chi_l^R u\| \\
	 &\lesssim \|(P_{\ve, \theta}-z-iK^R) u\|+\sum_l\|[P_{\ve, \theta}, \chi_l^R] u\|. 
\end{align*}
To estimate the commutator terms, 
we set $\chi_4^R(x)=\prod_{j=1}^N \tilde{\chi} (x_j/R), \chi_5^R(x)=\prod_{j=1}^N \tilde{\chi}(x_j/3R)$, 
where $\tilde{\chi}$ is the same as in the definition of $A$.    
Since $\chi_4^R=1$ on $\supp \partial \chi_l^R$ for $l=1, 2, 3$, we have 
\[
\sum_l\|[P_{\ve, \theta}, \chi_l^R] u\|\lesssim R^{-1}\|\chi_4^R u\|_{H^1}. 
\]
We now take $k=1$ and use Lemma~\ref{main-lemma}. 
Here we take $M=3R$ in the definition of $A$. 
Since $A^{-1}:L^2\to H^1$ is $\mathcal{O}(R^{-1/2})$, 
\begin{align*}
&R^{-1}\|\chi_4^R u\|_{H^1} \\
&\lesssim R^{-3/2}\|A\chi_4^R u\| \\
&=R^{-3/2}\|(P_{\ve, \theta}-z-iK^R-3iR)\chi_4^R u\|\\
&\lesssim R^{-3/2}\|(P_{\ve, \theta}-z-iK^R)u\|+R^{-1/2}\|u\|
+R^{-3/2}\|[P_{\ve, \theta}, \chi_4^R] u\|.
\end{align*}
Here we used $(1-\tilde{\chi}^{3R}(x_j))\chi_4^R(x)=0$ (see the definition of $A$).  
We then use Lemma~\ref{main-lemma} with $k=2$ and $M=7R$ in the definition of $A$. 
Since $A^{-1}:H^{-1}\to H^1$ uniformly, 
the last commutator term is bounded by 
\begin{align*}
&R^{-5/2}\|\chi_5^R u\|_{H^1} \\
&\lesssim R^{-5/2}\|A\chi_5^R u\|_{H^{-1}} \\
&=R^{-5/2}\|(P_{\ve, \theta}-z-iK^R-7iR)\chi_5^R u\|_{H^{-1}}\\
&\lesssim R^{-5/2}\|(P_{\ve, \theta}-z-iK^R)u\|+R^{-3/2}\|u\|
+R^{-5/2}\|[P_{\ve, \theta}, \chi_5^R] u\|_{H^{-1}}.
\end{align*}
Here we used $(1-\tilde{\chi}^{7R}(x_j))\chi_5^R(x)=0$. 
The last commutator term is bounded by $R^{-7/2}\|u\|$. 

Summing up these and taking $R\gg 1$, we conclude that 
\[
\|u\|\le C\|(P_{\ve, \theta}-z-iK^R)u\|. 
\]
By approximation, this is valid for all $u$ in the domain of $P_{\ve, \theta}$. 
Since we have the same estimate for the adjoint operator, we conclude that 
$\|(P_{\ve, \theta}-z-iK^R)^{-1}\|_{L^2\to L^2}\le C$ for $R\gg 1$. 
We fix such $R\gg 1$.

We finally prove $\|(P_{\ve, \theta}-z-iK^R)^{-1}\|_{L^2\to H^k(\{|x|<M\})}\le CM^{k/2}$. 
By the interpolation, it is enough to prove it for $k=2$ since the case where $k=0$ follows from the previous estimate.  
We may assume that $M\gg 1$.  
We take $M\gg 1$ and we use Lemma~\ref{main-lemma} with $k=2$, where $M$ is replaced with $7M$ and 
$R$ is the above fixed one.  
We then have  
\begin{align*}
&\|(P_{\ve, \theta}-z-iK^R)^{-1}\|_{L^2\to H^2(\{|x|<M\})}\\
&\le \|\chi_4^M (P_{\ve, \theta}-z-iK^R)^{-1}\|_{L^2\to H^2}\\
&\le \|\chi_4^M A^{-1} \chi_5^M A (P_{\ve, \theta}-z-iK^R)^{-1}\|_{L^2\to H^2}\\
&\,\,\,\,\,\,\,\,\,+\|\chi_4^M  A^{-1}(1-\chi_5^M) A (P_{\ve, \theta}-z-iK^R)^{-1}\|_{L^2\to H^2}. 
\end{align*}
Since $\chi_5^M(x)(1-\tilde{\chi}^{7M}(x_j))=0$ and $A^{-1}:L^2\to H^2$ uniformly for $M\gg 1$, 
the first term is bounded by 
\begin{align*}
\|\chi_5^M (P_{\ve, \theta}-z-iK^R-7iM) (P_{\ve, \theta}-z-iK^R)^{-1}\|\lesssim M. 
\end{align*}
Since $\chi_4^M(1-\chi_5^M)=0$, the second term is bounded by 
\begin{align*}
&\|\chi_4^MA^{-1}(1-\chi_5^M) A (P_{\ve, \theta}-z-iK^R)^{-1}\|_{L^2\to H^2}\\
&=\|A^{-1}[\chi_4^M, A]A^{-1}(1-\chi_5^M) A (P_{\ve, \theta}-z-iK^R)^{-1}\|_{L^2\to H^2}\\
&=\|A^{-1}[\chi_4^M, A]A^{-1}[\chi_5^M, A] (P_{\ve, \theta}-z-iK^R)^{-1}\|_{L^2\to H^2}\\
&\lesssim \|A^{-1}\|_{L^2\to H^2}\cdot \|[\chi_4^M, A]\|_{H^1\to L^2}\cdot 
\|A^{-1}\|_{H^{-1}\to H^1}\cdot \|[\chi_5^M, A]\|_{L^2\to H^{-1}}\\
&\lesssim M^{-2}.
\end{align*}
Summing up these, we have $\|(P_{\ve, \theta}-z-iK^R)^{-1}\|_{L^2\to H^2(\{|x|<M\})}\le CM$ and 
finish the proof of the proposition. 

\end{proof}

\section*{Acknowledgement}
The author is supported by JSPS KAKENHI Grant Number JP23KJ2090.

Research Organization of Science and Technology, Ritsumeikan University, 
1-1-1, Nojihigashi, Kusatsu-shi, Shiga, 525-8577, Japan

E-mail address: kameoka@gst.ritsumei.ac.jp

\end{document}